\documentclass[letterpaper,11pt]{article}

\usepackage[english]{babel}
\usepackage[utf8]{inputenc}
\usepackage{palatino}
\usepackage{amsmath, amssymb, amsthm, amsfonts}
\usepackage{mathrsfs}
\usepackage[
colorlinks=true,
linkcolor=blue,
anchorcolor=blue,
citecolor=blue,
urlcolor=blue,
plainpages=false,
pdfpagelabels
]{hyperref}
\usepackage{graphicx}
\usepackage[table]{xcolor}
\usepackage{wrapfig}

\usepackage[colorinlistoftodos,textsize=scriptsize]{todonotes}
\usepackage{fullpage}
\usepackage[compact]{titlesec}
\usepackage[margin=0.5in]{caption}
\usepackage{enumerate}
\usepackage{url}
\usepackage{pgf,tikz}
\usepackage[all]{xy}
\usepackage{tikz-cd}
\usetikzlibrary{calc}
\usetikzlibrary{decorations.pathreplacing}
\usetikzlibrary{arrows}
\usepackage{mathrsfs}  
\usepackage{blkarray}
\usepackage{subcaption}
\usepackage{etoolbox}
\usepackage{mathtools}

\usepackage[capitalise]{cleveref}
\usepackage{xspace}
\usepackage{xstring}
\usepackage[numbers,sort,compress]{natbib}
\theoremstyle{plain}
\newtheorem{theorem}{Theorem}[section]

\newtheorem{lemma}[theorem]{Lemma} 
\newtheorem{corollary}[theorem]{Corollary}

\theoremstyle{definition}
\newtheorem{example}[theorem]{Example}
\newtheorem{definition}[theorem]{Definition}

\theoremstyle{definition}
\newtheorem{remark}[theorem]{Remark}

\let\originalparagraph\paragraph
\renewcommand{\paragraph}[2][.]{\originalparagraph{#2#1}}

\DeclareMathOperator{\im}{\mathrm im}

\DeclareMathOperator{\coim}{\mathrm coim}

\DeclareMathOperator{\ob}{\mathrm Ob}

\DeclareMathOperator{\Reeb}{Reeb}

\newcommand{\QCat}{\mathbf{Q}}

\newcommand{\DL}{\mathchoice
{\begin{tikzpicture}
\draw (1.0ex,0ex)--(0ex,0ex) -- (0ex,1.0ex);
\end{tikzpicture}}
{\begin{tikzpicture}
\draw (1.0ex,0ex)--(0ex,0ex) -- (0ex,1.0ex);
\end{tikzpicture}}
{\begin{tikzpicture}
\draw[thick] (1.0ex,0ex)--(0ex,0ex) -- (0ex,1.0ex);
\end{tikzpicture}}
{\begin{tikzpicture}
\draw[line width=.6pt](.7ex,0ex)--(0ex,0ex) -- (0ex,.7ex);
\end{tikzpicture}}
}

\renewcommand{\phi}{\varphi}

\newcommand{\B}{\mathcal B}
\newcommand{\C}{\mathcal C}
\newcommand{\D}{\mathcal D}

\newcommand{\G}{\mathcal{G}}

\newcommand{\R}{\mathbb{R}}

\newcommand{\FS}{\mathcal S^\uparrow}
\newcommand{\W}{\mathcal{W}}

\newcommand{\Z}{\mathbb{Z}}

\newcommand{\RCat}{\mathbf{R}}
\newcommand{\ZCat}{\mathbf{Z}}
\newcommand{\PCat}{\mathbf{P}}

\newcommand{\CoEx}[0]{E}

\newcommand{\E}[1]{\ifthenelse{\equal{#1}{}}{E}{E(#1)}}
\newcommand{\FI}[1]{\mathcal S(#1)}
\newcommand{\Set}{\mathbf{Set}}
\newcommand{\idf}[1]{\emph{#1}}

\newcommand{\fl}{\mathrm{fl}}

\newcommand{\npGen}[1]{{#1{\kern .1ex}}^{\DL}}

\newcommand{\LSB}{\mathcal L}

\newcommand{\I}{{\mathcal I}}
\newcommand{\J}{{\mathcal J}}
\newcommand{\K}{{\mathcal K}}
\newcommand{\id}{{\rm id}}
\newcommand{\Top}{\mathbf{Top}}
\newcommand{\vect}{\mathbf{vec}}
\newcommand{\Vect}{\mathbf{Vec}}
\newcommand{\Hom}{{\rm Hom}}

\newcommand{\kk}{\mathbb{K}}

\newcommand{\REComp}[2]{\ifthenelse{\equal{#1}{}}{X_{#2}}{X_{#2}(#1)}}

\newcommand{\HH}{{\rm H}}
\renewcommand{\subseteq}{\subset}
\newcommand{\Op}{\mathcal{O}}
\renewcommand{\subseteq}{\subset}

\date{}

\title{Computational Complexity of the Interleaving Distance\footnote{M. B. Botnan has been supported by the DFG Collaborative Research Center SFB/TR 109 “Discretization in Geometry and Dynamics”. This work was partially carried out while the authors were visitors to the Hausdorff Center for Mathematics, Bonn, during the special Hausdorff program on applied and computational topology.}}
\author{Håvard Bakke Bjerkevik\thanks{Norwegian University of Science and Technology, Trondheim, Norway; \texttt{havard.bjerkevik@ntnu.no}} \and Magnus Bakke Botnan\thanks{TU M\"unchen, Munich, Germany; \texttt{botnan@ma.tum.de}}}

\begin{document}
\maketitle 
\begin{abstract}
The interleaving distance is arguably the most prominent distance measure in topological data analysis. In this paper, we provide bounds on the computational complexity of determining the interleaving distance in several settings. We show that the interleaving distance is NP-hard to compute for persistence modules valued in the category of vector spaces. In the specific setting of multidimensional persistent homology we show that the problem is at least as hard as a matrix invertibility problem. Furthermore, this allows us to conclude that the interleaving distance of interval decomposable modules depends on the characteristic of the field. Persistence modules valued in the category of sets are also studied. As a corollary, we obtain that the isomorphism problem for Reeb graphs is graph isomorphism complete.
\end{abstract}
\section{Introduction} 
For a category $\C$ and a poset $\PCat$ we define a \emph{$\PCat$-indexed (persistence) module valued in $\C$} to a be a functor $M: \PCat\to \C$. We will denote the associated functor category by $\C^\PCat$. If $M,N\in \C^\PCat$ then $M$ and $N$ are of the same \emph{type}. Such functors appear naturally in applications, and most commonly when $\PCat=\RCat^n$, $n$-tuples of real numbers under the normal product order, and $\C=\Vect_{\kk}$, the category of vector spaces over the field $\kk$, or $\C=\Set$, the category of sets. The field $\kk$ is assumed to be finite. We suppress notation and simply write $\Vect$ when $\kk$ is an arbitrary finite field. The notation $p\in \PCat$ denotes that $p$ is an object of $\PCat$. 
\begin{remark}
Throughout the paper we make use of basic concepts from category theory. The reader unfamiliar to such ideas will find the necessary background material in the first few pages of \cite{bubenik2014categorification}.
\end{remark}

Assume that $h: X\to \R$ is a continuous function of ``Morse type'', a generalization of a Morse function on a compact manifold. Roughly, a real-valued function is of Morse type if the homotopy type of the fibers changes at finite set of values; see \cite{carlsson2009zigzag} for a precise definition. We shall now briefly review four different scenarios in which functors of the aforementioned form can be associated to $h$. 

Let $\HH_p: \Top \to \Vect_\kk$ denote the $p$-th singular homology functor with coefficients in $\kk$, and let $\pi_0: \Top\to \Set$ denote the functor giving the set of path-components. We also associate the two following functors to $h$ whose actions on morphisms are given by inclusions: 
\begin{align*}
\FS(h): \RCat\to \Top & \quad & \FI{h}: \RCat^2\to \Top\\
\FS(h)(t) = \{x\in X \mid h(x) \leq t\} & \quad & \FI{h}(-s,t) = \{x\in X \mid s\leq h(x) \leq t\}
\end{align*}
\begin{itemize}
\item \emph{Persistent Homology} studies the evolution of the homology of the sublevel sets of $h$ and is perhaps the most prominent tool in topological data analysis \cite{carlsson2009zigzag}. Specifically, the \emph{$p$-th sublevel set persistence module associated to $h$} is the functor $\HH_p\FS(h):\RCat\to \Vect$. Importantly, such a module is completely determined by a collection of intervals $\B(\HH_p\FS(h))$ called the \emph{barcode} of $\HH_p\FS(h)$. This collection of intervals is then in turn used to extract topological information from the data at hand. In \cref{MorseBlob} we show the associated barcode for $p=0$ and $p=1$ for a function of Morse type. 

\item Upon replacing $\HH_p$ by $\pi_0$ in the above construction we get a \emph{merge tree}. That is, the \emph{merge tree associated to $h$} is the functor $\tau^h: \pi_0 \FS(h): \RCat \to \Set$. A merge tree captures the evolution of the path components of the sublevel sets of $h$ and can be, as the name indicates, be visualized as (a disjoint union of) rooted trees. See \cref{MorseBlob} for an example. 
\item 
The two aforementioned examples used sublevel sets. A richer invariant is obtained by considering interlevel sets: define the \emph{$p$-th interlevel set persistence of $h$} to be the functor $\HH_p \FI{h}: \RCat^2\to \Vect$. Analogously to above, such a module is completely determined by a collection $\B(\HH_p\FI{h})$ of simple regions in $\R^2$. However, it is often the collection of intervals $\LSB_p(h)$ obtained by the intersection of these regions with the anti-diagonal $y=-x$ which are used in data analysis. We refer the reader to \cite{botnan2016algebraic} for an in-depth treatment. In \cref{MorseBlob} we show an example of the $0$-th interlevel set barcode. Observe how the endpoints of the intervals correspond to different types of features of the Reeb graph. 

\item
Just as interlevel set persistence is a richer invariant than sublevel set persistence, the \emph{Reeb graph} is richer in structure than the merge tree. Specifically, we define the functor $\Reeb^h:= \pi_0\FI{h}: \RCat^2\to \Set$. Just as for Merge trees, $\Reeb^h$ admits a visualization of a graph; see \cref{MorseBlob}. In particular, this appealing representation has made Reeb graphs a popular objects of study in computational geometry and topology, and they have found many applications in data visualization and exploratory data analysis.
\end{itemize}
\begin{figure}
\centering
\begin{tikzpicture}[scale=.35]
\begin{scope}[rotate=90]
\begin{scope}
\fill[red!60,draw=black,even odd rule]
(0,0) to [out=90,in=180] (3,1.5) to [out=0,in=270] (7,1.5) to [out=90,in=270] (4,2) to [out=90,in=90] (9,1) to [out=270,in=90] (6,-.5) to [out=270,in=90] (8,-1) to [out=270,in=0] (4.5,-1.5) to [out=180,in=270] (1,-1) to [out=90,in=270] (2,0) to [out=90,in=270] (0,0)
(3,-.2) to [out=90, in=90] (5,-.2) to [out=270, in=270] (3,-.2);
\draw[very thick, ->] (-1,4) -- (10,4);
\def\x{-26}
\draw[dotted] (0,0) -- (0,\x);
\draw[dotted] (1,-1) -- (1,\x);
\draw[dotted] (2,0) -- (2,\x);
\draw[dotted] (3,-.2) -- (3,\x);
\draw[dotted] (4,2) -- (4,\x);
\draw[dotted] (5,-.2) -- (5,\x);
\draw[dotted] (6,-.5) -- (6,\x);
\draw[dotted] (7,1.5) -- (7,\x);
\draw[dotted] (8,-1) -- (8,\x);
\draw[dotted] (9,1) -- (9,\x);
\end{scope}
\begin{scope}[yshift=-5cm]
\node at (-1,0){${\Reeb^h}$};
\draw[thick] (0,.5) to (2,0);
\draw[thick] (1,-.5) to (2,0);
\draw[thick] (2,0) to (3,0);
\draw[thick] (3,0) to [out=70,in=110] (5,0);
\draw[thick] (3,0) to [out=290,in=250] (5,0);
\draw[thick] (4,1.5) to (7,1);
\draw[thick] (5,0) to (6,0);
\draw[thick] (6,0) to (7,1);
\draw[thick] (7,1) to (9,1.0);
\draw[thick] (6,0) to (8,-.5);
\end{scope}
\begin{scope}[yshift=-8.5cm]
\node at (-1,0){ $\tau^h$};
\draw[thick] (0,.5) to (2,0);
\draw[thick] (1,-.5) to (2,0);
\draw[thick] (2,0) to (7,.5);
\draw[thick] (4,1.5) to (7,.5);
\draw[thick, ->] (7,.5) to (10,.5);
\end{scope}
\begin{scope}[yshift=-14cm]
\node at (-1,0){$\B(H_0\FS(h))$};
\draw[thick, [->] (0,.5) to (10,.5);
\draw[thick, [-)] (1,0) to (2,0);
\draw[thick, [-)] (4,-.5) to (7,-.5);
\end{scope}
\begin{scope}[yshift=-20cm]
\node at (-1,0){$\B(H_1\FS(h))$};
\draw[thick, [->] (5,0) to (10,0);
\end{scope}
\begin{scope}[yshift=-26cm]
\node[text width=3cm,align=center] at (-1,0){\small $\LSB_0(h)$};
\draw[thick, [-)] (1,2) to (2,2);
\draw[thick, (-)] (3,1.5) to (5,1.5);
\draw[thick, [-)] (4,1) to (7,1);
\draw[thick, [-)] (8,0.5) to (6,0.5);
\draw[thick, {[-]}] (0,0) to (9,0);

\end{scope}
\end{scope}
\end{tikzpicture}
\caption{The height function of the solid shape is of Morse type. The associated Reeb graph, merge tree, sublevel set barcodes, and interlevel set barcode are shown to the right. \label{MorseBlob}}
\end{figure}
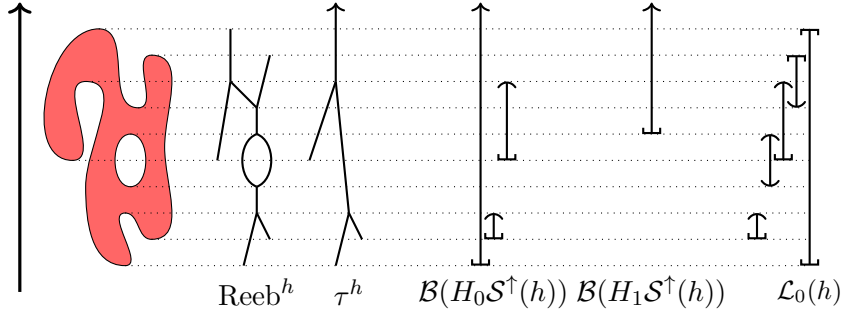

These are all examples of topological invariants arising from a single real-valued function. There are many settings for which it is more fruitful to combine a collection of real-valued functions into a single function $g: X\to \R^n$ \cite{carlsson2009theory}. By combining them into a single function we not only learn how the data looks from the point of view of each function (i.e. a type of measurement) but how the different functions (measurements) interact. One obvious way to assign a (algebraic) topological invariant to $g$ is to filter it by sublevel sets. That is, define $\FS(g): \RCat^n \to \Top$ by $\FS(g)(t) = \{x\in X \mid g(x) \leq t\}.$ The associated functor $\HH_p\FS(g): \RCat^n\to \Vect$ is an example of an  \emph{$n$-dimensional persistence module}. We saw above that for $n=1$ this functor is completely described by a collection of intervals. This is far from true for $n\geq 2$: there exists no way to describe such functors by interval-like regions in higher-dimensional Euclidean space. Even the task of parameterizing such (indecomposable) modules is known to be a \emph{hopeless} problem (so-called \emph{wild representation type}) \cite{barot2014introduction}.

\subsection{The Interleaving Distance}
Different types of distances have been proposed on various types of persistence modules with values in $\Vect$ \cite{scolamiero2017multidimensional, lesnick2015theory,chazal2009proximity, bubenik2014metrics, biasotti2008multidimensional}. Of all these, the \emph{interleaving distance} is arguably the most prominent for the following reasons: the theory of interleavings lies at the core of the theoretical foundations of 1-dimensional persistence, notably through the Isometry Theorem (\cref{teo:isometry}). Furthermore, it was shown by Lesnick that when $\kk$ is a prime field, the interleaving distance is the \emph{most discriminative} of all \emph{stable} metrics on such modules. We refer to \cite{lesnick2015theory} for the precise statement. As we shall see, it is also an immediate consequence of \cref{teo:isometry} that the interleaving distance for 1-dimensional persistence modules can be computed in polynomial time.

Lesnick's result generalizes to $n$-dimensional persistence modules, but the computational complexity of computing the interleaving distance of such modules remains unknown. An efficient algorithm to compute the interleaving distance could carry a profound impact on topological data analysis: the standard pipeline for $1$-dimensional persistent homology is to first compute the barcode and then perform analysis on the collection of intervals. However, for multi-dimensional persistence there is no way of defining the barcode. With an efficient algorithm for computing the interleaving distance at hand it would still not be clear how to analyze the persistence modules individually, but we would have a theoretical optimal way of comparing them. This in turn could be used in clustering, kernel methods, and other kinds of data analysis widely applied in the 1-dimensional setting. 
%

\subsection*{Complexity}
The purpose of this paper is to determine the computational complexity of computing the interleaving distance. To make this precise, we need to associate a notion of size to the persistence modules. 
\begin{definition}Let $\PCat$ denote a poset category and $M: \PCat\to \C$. 
\begin{itemize}
\item For $\C=\Vect$, define the \emph{total dimension of $M$} to be $\dim M = \sum_{p\in \PCat} \dim M_p$. 
\item For $M: \ZCat \to \Set$, define the \emph{total cardinality of $M$} to be $|M| = \sum_{p\in \PCat} |M_p|$. 
\end{itemize}
\end{definition}

The input size will be the total dimension or the total cardinality and for the the remaining of the paper \textbf{we shall always assume} that those quantities are finite. 
The following shows that there exists an algorithm, polynomial in the input size, which determines whether or not two $\PCat$-indexed modules valued in $\Vect$ are isomorphic. 

\begin{theorem}[\cite{brooksbank2008testing}]
Let $\PCat$ be a finite poset and $M, M': \PCat\to \Vect$. There exists a deterministic algorithm which decides if $M\cong M'$ in $\Op\left((\dim M + \dim M')^6\right)$. 
\label{thm:brooksbank}
\end{theorem}

This result will be important to us in what ensues because the strongest of interleavings, the 0-interleaving, is by definition a pair of inverse isomorphisms. Furthermore, by choosing an appropriate basis for each vector space, an isomorphism between $M$ and $M'$ is nothing more than a collection of matrices with entries in a finite field. Likewise a $\delta$-interleaving will be nothing more than a collection of matrices over a finite field satisfying certain constraints. When $\C=\Set$ the morphisms are specified by collections of functions between finite sets. Hence, the decision problems considered in this paper \textbf{are trivially in NP}.

Furthermore, it is an immediate property of the Morse type of $h$, that the modules considered above are \emph{discrete}. Intuitively, we say that an $\RCat^n$-indexed persistence module $M$ is discrete if there exists a $\ZCat^n$-indexed persistence module containing all the information of $M$; see \cref{app:discrete}. In practice, persistence modules arising from data will be discrete. Hence, when it comes to algorithmic questions we shall restrict ourselves to the setting in which $\PCat = \ZCat^n$ or a slight generalization thereof. Importantly, the modules considered in this paper can be $\delta$-interleaved only for $\delta\in \{0, 1, 2, \ldots\}$. 

\subsection*{Contributions}
\label{contr}
The contributions of this paper are summarized in \cref{table:contributions}. Concretely, a cell in \cref{table:contributions} gives a complexity bound on the decision problem of deciding if two modules of the given type are $\delta$-interleaved. It is an easy consequence of the definition of the interleaving distance that this is at least as hard as determining the distance itself. The cells with a shaded background indicate that novel contributions to that complexity bound is provided in this paper. Recall that we have defined the input size to be $n = \dim M + \dim M'$ when the modules are valued in $\Vect$, and $n = |M| + |M'|$ when the modules are valued in $\Set$.  Observe that any non-trivial functor $M: \ZCat^m\to \Set$ must have $|M| = \infty$. Hence, when we talk about interleavings of such functors, we shall assume that they are completely determined by a restriction to a finite sub-grid. The input size is then the total cardinalities of the restrictions. We will now give a brief summary of the cells of \cref{table:contributions}. 

\begin{itemize}
\item $\ZCat\to \Vect$. [$\delta\geq 0$] This bound is achieved by first determining the barcodes of the persistence modules and then using \cref{teo:isometry} to obtain the interleaving distance. The complexity of this is $\Op({\rm Find Barcode} + {\rm Match}) = \Op(n^\omega + n^{1.5}\log n) = \Op(n^\omega)$ where $\omega$ is the matrix multiplication exponent\cite{kerber2017geometry}. The details can be found in \cref{app:vec}. In \cite{milosavljevic2011zigzag}, the complexity is shown to be $\Op(n^\omega + n^2 \log^2 n)$ for essentially the same problem, but with a slightly different input size $n$.
\item $\ZCat\to \Set$. [$\delta=0$] Essentially isomorphism of rooted trees; see \cref{app:mergetree}. [$\delta\geq 1$] This follows from arguments in \cite{agarwal2015computing}.
\item $\ZCat^2\to \Vect$. [$\delta=0$] This is \cref{thm:brooksbank} for $\PCat=\ZCat^2$. [$\delta\geq 1$] A \idf{constrained invertibility} (CI) problem is a triple $(P,Q,n)$ where $P$ and $Q$ are subsets of $\{1,2, \dots, n\}^2$. We say that a CI-problem $(P,Q,n)$ is \idf{solvable} if there exists an invertible $n \times n$ matrix $M$ such that $M_{i,j} = 0$ for all $(i,j) \in P$ and $M^{-1}_{i',j'} = 0$ for all $(i',j') \in Q$. We call $(M,M^{-1})$ a \idf{solution} of $(P,Q,n)$. In \cref{sec:multid} we show that a  CI-problem is solvable if and only if an associated pair of $\ZCat^2$-indexed modules is $1$-interleaved. Thus, the interleaving problem is \emph{constrained invertibility-hard} (CI-hard).
\item $\ZCat^2\to \Set$. [$\delta=0$] Reeb graphs are a particular type of functors $\ZCat^2\to \Set$ and deciding if two Reeb graphs are isomorphic is graph isomorphism-hard (GI-hard) \cite{de2015categorified}. In \cref{app:isoGI} we strengthen this result by showing that the isomorphism problem for $\ZCat^2\to \Set$ is in fact GI-\emph{complete}. This also implies that Reeb graph isomorphism is GI-complete. [$\delta\geq 1$] This follows from $\ZCat\to \Set$. 
\item $\ZCat^{L,C}\to \Vect_{\Z/2\Z}$.
For two sets $L$ and $C$, define $\ZCat^{L\to C}$ to be the poset generated by the following disjoint union of posets
$\ZCat^{L\to C} := \bigsqcup_{l\in L, c\in C} \ZCat$
with the added relation $(l,t) < (c,t)$ for every $l\in L$, $c\in C$ and $t\geq 3$.
This poset is a mild generalization of a disjoint union of $\ZCat$'s. [$\delta=0$] Immediate from \cref{thm:brooksbank}. [$\delta \geq 1$] Follows from a reduction from 3-SAT; see \cref{sec:NPHard}.  This shows that computing the generalized interleaving distance of \cite{bubenik2014metrics} for $\Vect$-valued persistence modules is NP-complete in general. 
\end{itemize}
\begin{table}
\begin{center}
\begin{tabular}{ l | c | c | c | c | c  }
type/$\delta$& $\ZCat\to \Vect$ & $\ZCat\to \Set$  & $\ZCat^2\to \Vect$  & $\ZCat^2\to \Set$ & $\ZCat^{L,C} \to \Vect_{\Z/2\Z}$ \\ 
$\delta=0$ & $\Op(n^\omega)$ & \cellcolor{gray!25} $\Op(n)$ & $\Op(n^6)$  & \cellcolor{gray!25} GI-complete & $\Op(n^6)$\\ 
$\delta\geq 1$ & $\Op(n^\omega)$ & NP-complete & \cellcolor{gray!25} CI-hard & NP-complete  & \cellcolor{gray!25}NP-complete\\
\end{tabular}
\caption{The complexity of checking for $\delta$-interleavings between modules $M$ and $M'$. If the target category is $\Vect$ then $n = \dim M + \dim M'$, and if the target category is $\Set$ then $n = |M| + |M'|$. Here $\omega$ is the matrix multiplication exponent. }
\label{table:contributions}
\end{center}
\end{table}

\section{Preliminaries}
For $\PCat$ a poset and $\C$ an arbitrary category, $M:\PCat\to \C$ a functor, and $a,b\in \PCat$, let $M_a= M(a)$, and let $\phi_M(a,b) : M_a \to M_b$ denote the morphism $M(a\leq b)$. 
\subsection{Interleavings} 
\label{sec:interleaving}
In this section we review the theory of interleavings for $\ZCat^n$-indexed modules. For a treatment of the $\RCat^n$-indexed setting see \cite{lesnick2015theory}. For a discussion on interleavings over arbitrary posets see \cite{bubenik2014metrics}. 

For $u\in \ZCat^n$, define the \emph{$u$-shift functor} $(-)(u): \C^{\ZCat^n} \to \C^{\ZCat^n}$ on objects by $M(u)_a = M_{u+a}$, together with the obvious internal morphisms, and on morphisms $f: M\to N$ by $f(u)_a = f(u+a): M(u)_a \to N(u)_a$. For $u\in \{0, 1, \ldots\}^n$, let $\phi_M^u: M \to M(u)$ be the morphism whose restriction to each $M_a$ is the linear map $\phi_M(a, a+u)$. For $\delta\in \{0,1,2\ldots\}$ we will abuse notation slightly by letting $(-)(\delta)$ denote the $\delta(1, \ldots, 1)$-shift functor, and letting $\phi_M^\delta$ denote $\phi^{\delta(1, \ldots, 1)}_M$.

\begin{definition}
Given $\delta\in \{0,1, \ldots\}$, a $\delta$-interleaving between $M,N:\ZCat^n\to \C$ is a pair of morphisms $f: M\to N(\delta)$ and $g: N\to M(\delta)$ such that 
$g(\delta)\circ f = \phi_M^{2\delta}$ and $f(\delta)\circ g = \phi_N^{2\delta}.$
\label{def:interleaving}
\end{definition}
We call $f$ and $g$ \emph{$\delta$-interleaving morphisms}.  If there exists a $\delta$-interleaving between $M$ and $N$, we say $M$ and $N$ are $\delta$-interleaved.
The \emph{interleaving distance} $d_I: \ob(\C^{\ZCat^n})\times\ob(\C^{\ZCat^n})\to [0,\infty]$ is given by $d_I(M,N) = \min \{\delta\in \{0, 1, \ldots\} \mid \text{$M$ and $N$ are $\delta$-interleaved}\}.$ Here we set $d_I(M,N) = \infty$ if there does not exist a $\delta$-interleaving for any $\delta$. 
\label{sec:interleaving}
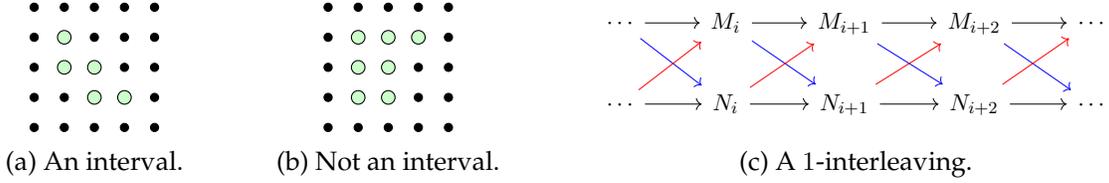
\begin{figure}
\centering
\begin{subfigure}[t]{0.22\textwidth}
\centering
\begin{tikzpicture}[scale=1]
\foreach \x in {0.4,0.8,1.2,1.6,2.0}
\foreach \y in {0.4,0.8,1.2,1.6,2.0}
{
\draw[fill=black] (\x,\y) circle (.3ex);
}
\draw[black,fill=green!20] (0.8,1.6) circle (.5ex);
\draw[black,fill=green!20] (1.2,0.8) circle (.5ex);
\draw[black,fill=green!20] (0.8,1.2) circle (.5ex);
\draw[black,fill=green!20] (1.2,1.2) circle (.5ex);
\draw[black,fill=green!20] (1.6,0.8) circle (.5ex);
\end{tikzpicture}
\subcaption{An interval.}
\end{subfigure}
\hskip5pt
\begin{subfigure}[t]{0.22\textwidth}
\centering
\begin{tikzpicture}[scale=1]
\foreach \x in {0.4,0.8,1.2,1.6,2.0}
\foreach \y in {0.4,0.8,1.2,1.6,2.0}
{
\draw[fill=black] (\x,\y) circle (.3ex);
}
\draw[black,fill=green!20] (0.8,0.8) circle (.5ex);
\draw[black,fill=green!20] (0.8,1.2) circle (.5ex);
\draw[black,fill=green!20] (0.8,1.6) circle (.5ex);
\draw[black,fill=green!20] (1.2,1.6) circle (.5ex);
\draw[black,fill=green!20] (1.6,1.6) circle (.5ex);
\draw[black,fill=green!20] (1.2,1.2) circle (.5ex);
\draw[black,fill=green!20] (1.2,0.8) circle (.5ex);
\end{tikzpicture}
\subcaption{Not an interval. }
\end{subfigure}
\hskip5pt
\begin{subfigure}[t]{0.5\textwidth}
\centering
\begin{tikzpicture}[scale=1][baseline= (a).base]
\node[scale=0.8] (a) at (0,0){
\begin{tikzcd}
\cdots\ar[r]\ar[dr, color=blue] & M_i \ar[r]\ar[dr, color=blue] &M_{i+1}\ar[r]\ar[dr, color=blue] & M_{i+2}\ar[r]\ar[dr, color=blue] &\cdots\\
\cdots\ar[r]\ar[ur,color=red]& N_i\ar[r]\ar[ur, color=red]&N_{i+1}\ar[r]\ar[ur, color=red] &N_{i+2}\ar[r]\ar[ur, color=red] &\cdots
\end{tikzcd}};
\end{tikzpicture}
\subcaption{A $1$-interleaving. }
\end{subfigure}
\caption{(a) is an interval in $\ZCat^2$ whereas (b) is not. (c) The persistence modules $M$ and $N$ are 1-interleaved if and only if there exist diagonal morphisms such that the diagram in (c) commutes.}
\label{fig:gridRestriction}
\end{figure}
\subsection{Interval Modules and the Isometry Theorem}
Let $\C = \Vect$. \label{Sec:Barcodes}
An \emph{interval} of a poset $\PCat$ is a subset $\J\subseteq \PCat$ such that 
\begin{enumerate}
\item $\J$ is non-empty.
\item If $a,c\in \J$ and $a\leq b\leq c$, then $b\in \J$.
\item {}[connectivity] For any $a,c\in \J$, there is a sequence $a=b_0,b_1,\ldots, b_l=c$ of elements of $\J$ with $b_i$ and $b_{i+1}$ comparable for $0\leq i\leq l-1$.   
\end{enumerate}
We refer to a collection of intervals in $\PCat$ as a \emph{barcode (over $\PCat$)}.  

\begin{definition}
\label{def:interval}
For $\J$ an interval in $\PCat$, the interval module $I^\J$ is the $\PCat$-indexed module such that
\begin{align*}
I^\J_a&=
\begin{cases}
\kk &{\textup{if }} a\in \J, \\
0 &{\textup{ otherwise}.}
\end{cases}
& \varphi_{I^\J}(a,b)=
\begin{cases}
\id_\kk &{\textup{if }} a\leq b\in I,\\
0 &{\textup{ otherwise}.}
\end{cases}
\end{align*}
\end{definition}
We say a persistence module $M$ is \emph{decomposable} if it can be written as $M\cong V\oplus W$ for non-trivial persistence modules $V$ and $W$; otherwise, we say that $M$ is \emph{indecomposable}. 

A $\PCat$-indexed module $M$ is \emph{interval decomposable} if there exists a collection $\B(M)$ of intervals in $\PCat$ such that 
$M\cong \bigoplus_{\J\in \B(M)} I^{\J}.$
We call $\B(M)$ the \emph{barcode} of $M$. This is well-defined by the Azumaya--Krull--Remak--Schmidt theorem \cite{azumaya1950corrections}. 

\begin{theorem}[Structure of 1-D Modules \cite{crawley2012decomposition,webb1985decomposition}]\label{Structure_Theorem}
Suppose $M: \PCat\to \Vect$ for $\PCat\in \{\RCat,\ZCat\}$ and $\dim M_p<\infty$ for all $p\in \PCat$. Then $M$ is interval decomposable.
\end{theorem}
\begin{remark}
Such a decomposition theorem exists only for \emph{very} special choices of $\PCat$. Two other scenarios appearing in applications are \emph{zigzags} \cite{botnan2015interval, carlsson2010zigzag} and exact bimodules \cite{cochoy2016decomposition}. The latter is a specific type of $\RCat^2$-indexed persistence modules. \label{rem:zigzag}
\end{remark}
\begin{corollary}
Let $\PCat=\bigsqcup_{i\in \Lambda} \ZCat$ be the poset given as a disjoint union of $\ZCat$'s (i.e. elements in different components are incomparable). If $M: \PCat\to \Vect$ satisfies $\dim M_p < \infty$ for all $p\in \PCat$, then $M$ is interval decomposable.
\label{cor:disjoint}
\end{corollary}
\begin{proof}Apply \cref{Structure_Theorem} to each of the components of $\PCat$ independently. This gives  $\B(M) = \bigsqcup_{i\in \Lambda} \B(M|_{(i,\ZCat)}).$ \end{proof}
At the very core of topological data analysis are the \emph{isometry theorems}. They say that for certain choices of interval decomposable modules, the interleaving distance coincides with a completely combinatorial distance on their associated barcodes. This combinatorial distance $d_B$ is called the \emph{bottleneck distance} and is defined in \cref{sec:bottleneck}. Importantly, for any two barcodes, if the interleaving distance between each pair of interval modules in the barcodes is known, the associated bottleneck distance can be computed by solving a bipartite matching problem. This, in turn, implies that the interleaving distance can be efficiently computed whenever an isometry theorem holds. See \cref{app:vec} for an example. 
\begin{theorem}[Isometry Theorem \cite{lesnick2015theory, chazal2009proximity, bauer2015induced, bjerkevik2016stability}]
\label{teo:isometry}
Suppose $M,N: \ZCat\to \Vect$ satisfy $\dim M_i<\infty$ and $\dim N_i<\infty$ for all $i\in \ZCat$. Then $d_I(M,N) = d_B(\B(M), \B(N))$.
\end{theorem}
\begin{remark}
Continuing on the remark to \cref{Structure_Theorem}. An isometry theorem also holds for zigzags and exact bimodules \cite{botnan2016algebraic,bjerkevik2016stability}. Although there might be other classes of interval decomposable modules for which an isometry theorem holds, the result is not true in general. See \cref{app:unstable} for an example of interval decomposable modules in $\ZCat^2$ for which $2d_I(M,N) = d_B(\B(M), \B(N))$, and see \cite{botnan2016algebraic} for a general conjecture. This shows that a matching of the barcodes will not determine the interleaving distance even in the case of very well-behaved modules. 
\end{remark}
\section{NP-completeness}
\label{sec:NPHard}
In this section we shall prove that it is NP-hard to decide if two modules $M,N\in \Vect^{\ZCat^{L\to C}}$ are 1-interleaved. Recall that for two sets $L$ and $C$, we define $\ZCat^{L\to C}$ to be the disjoint union $\bigsqcup_{l\in L, c\in C} \ZCat$ with the added relations $(l,t) < (c,t)$ for all $l\in L, c\in C$, and $t\geq 3$. Define the \emph{$u$-shift functor} $(-)(u): \C^{\ZCat^{L\to C}} \to \C^{\ZCat^{L\to C}}$ on objects by $M(u)_{(p,t)} = M_{(p,t+u)}$, together with the obvious internal morphisms, and on morphisms $f: M\to N$ by $f(u)_{(p,t)} = f_{(p,t+u)}: M(u)_{(p,t)} \to N(u)_{(p,t)}$. That is, the shift functor simply acts on each of the components independently. With the shift-functor defined, we define a $\delta$-interleaving of $\ZCat^{L\to C}$-indexed modules precisely as in \cref{sec:interleaving}. Thus, we see that a $\delta$-interleaving is simply a collection of $\delta$-interleavings over each disjoint component of $\ZCat$ which satisfy the added relations. Indeed, a $1$-interleaving is equivalent to the existence of dashed morphisms in the following diagram for all $l\in L$ and $c\in C$:
$$
\begin{tikzcd}
\vdots & \vdots  & & \vdots  & \vdots \\
M_{(l,5)}\ar[rrr, bend left=20]\ar[u]\ar[ur, dashed] & N_{(l, 5)}\ar[rrr, bend left=20]\ar[ul, dashed]\ar[u] & & M_{(c, 5)}\ar[ur, dashed]\ar[u] & \ar[u]N_{(c,5)}\ar[ul, dashed] \\
M_{(l, 4)}\ar[ur, dashed]\ar[rrr, bend left=20]\ar[u] & N_{(l, 4)}\ar[u]\ar[rrr, bend left=20]\ar[ul, dashed]  & & M_{(c, 4)}\ar[ur, dashed]\ar[u] & N_{(c,4)}\ar[u]\ar[ul, dashed]  \\
M_{(l, 3)}\ar[ur, dashed]\ar[rrr, bend left=20]\ar[u] & N_{(l, 3)}\ar[u]\ar[rrr, bend left=20]\ar[ul, dashed]  & & M_{(c, 3)}\ar[ur, dashed]\ar[u] & N_{(c,3)}\ar[u]\ar[ul, dashed]  \\
M_{(l, 2)}\ar[ur, dashed]\ar[u] & N_{(l, 2)}\ar[u]\ar[ul, dashed]  & & M_{(c, 2)}\ar[ur, dashed]\ar[u] & N_{(c,2)}\ar[u] \ar[ul, dashed] \\
\vdots\ar[ur, dashed]\ar[u] & \vdots\ar[u]\ar[ul, dashed] & & \vdots\ar[ur, dashed]\ar[u] & \vdots\ar[u] \ar[ul, dashed]
\end{tikzcd}
$$
We saw in \cref{cor:disjoint} that $M: \ZCat^{L\to \emptyset}\to \Vect$ is interval decomposable. By applying \cref{teo:isometry} to each disjoint component independently, the following is easy to show. Here the bottleneck distance is generalized in the obvious way, i.e. matching each component independently. 
\begin{corollary}[Isometry Theorem for Disjoint Unions]
Let $L$ be any set, and $M,N: \ZCat^{L\to \emptyset} \to \Vect$ such that $\dim M_p<\infty$ and  $\dim N_p <\infty$ for all $p\in \ZCat^{L\to \emptyset}$. Then $d_I(M,N) = d_B(\B(M), \B(N))$. 
\end{corollary}
In particular, the interleaving distance between $M$ and $N$ can be effectively computed through a bipartite matching. As we shall see, this is not true for $C\neq \emptyset$.  The remainder of this section is devoted to proving the following theorem:
\begin{theorem}
\label{thm:NPHARD}
Unless P=NP, there exists no algorithm, polynomial in $n=\dim M + \dim N$, which decides  if $M,N: \ZCat^{L\to C}\to \Vect_{\Z/2\Z}$ are $1$-interleaved.
\end{theorem}
\subsection{The Proof}
We shall prove \cref{thm:NPHARD} by a reduction from 3-SAT. Let $\psi$ be a boolean formula in 3-CNF defined on literals $L = \{x_1, x_2, \ldots, x_{n_l}\}$ and clauses $C= \{c_1, c_2, \ldots, c_{n_c}\}$. We shall assume that the literals of each clause are distinct and \emph{ordered}. That is, the clause $c_i$ is specified by the three distinct literals $\{x_{i_1}, x_{i_2}, x_{i_3}\}$ wherein $i_1 < i_2 < i_3$. Determining if $\psi$ is satisfiable is well-known to be NP-complete. For the entirety of the proof $\kk = \Z/2\Z$.

\textbf{Step 1: Defining the representations.} Associate to $\psi$ two functors $M, N: \ZCat^{L\to C} \to \Vect_{\Z/2\Z}$ in the following way:
For all literals $x_j\in L$ define  
$$
\begin{tikzcd}
M_{(x_j,1)} \ar[r] &M_{(x_j,2)} \ar[r] &M_{(x_j,3)} \ar[r] &M_{(x_j,4)}=  \kk\ar[r, "1"] & \kk\ar[r, "1"] & \kk\ar[r, "1"] & \kk \\
N_{(x_j,1)} \ar[r] &N_{(x_j,2)} \ar[r] &N_{(x_j,3)} \ar[r] &N_{(x_j,4)} = \kk\ar[r, "1"] & \kk\ar[r, "(1;1)"] & \kk^2\ar[r] & 0
\end{tikzcd}
$$ 
and for every clause $c_i$ in $\psi$ define 
$$
\begin{tikzcd}
M_{(c_i,1)} \ar[r] &M_{(c_i,2)} \ar[r] &M_{(c_i,3)} \ar[r] &M_{(c_i,4)}=  0\ar[r] & \kk\ar[r, "1"] & \kk\ar[r, "1"] & \kk \\
N_{(c_i,1)} \ar[r] &N_{(c_i,2)} \ar[r] &N_{(c_i,3)} \ar[r] &N_{(c_i,4)} = 0\ar[r] & \kk^3\ar[r, "1"] & \kk^3 \ar[r] & 0
\end{tikzcd}
$$
For any other $p\in \ZCat^{L\to C}$, $M_p = N_p = 0$. Next we specify the remaining non-trivial morphisms: let $c_i=z_{i_1}\vee z_{i_2} \vee z_{i_3}$ be a clause in $\psi$, where $z_{j_l} = x_{j_l}$ or $z_{j_l} = \neg x_{j_l}$, and for which $i_1 < i_2 < i_3$. For $s=1,2,3$ define $H_s: \kk^2 \to \kk^3$ by $e_1 \to u\cdot e_s$ and $e_2 \to (1-u)\cdot e_s$, where $u=1$ if $z_{i_s} = x_{i_s}$ and $u=0$ if $z_{i_s} = \neg x_{i_s}$.  Here $e_d$ is the $d$-th standard basis vector of $\kk^3$. Given this we define the following for $s=1,2,3$: 
$$
\begin{tikzcd}
M_{(x_{i_s},4)}\ar[r] & M_{(c_i, 4)} & = & \kk\ar[r, "1"] & \kk \\
M_{(x_{i_s}, 3)}\ar[r]\ar[u] & M_{(c_i, 3)}\ar[u] & = & \kk\ar[r, "1"]\ar[u, "1"] & \kk\ar[u, "1"]
\end{tikzcd}
$$
and 
$$
\begin{tikzcd}
N_{(x_{i_s}, 3)}\ar[r] & N_{(c_i, 3)} =  \kk^2\ar[r, "H_s"]& \kk^3.
\end{tikzcd}
$$
Clearly $\dim M+\dim N=\Op(n_c+n_l)$. Thus, the total dimension is polynomial in the input size of 3-SAT. 

\textbf{Step 2: Showing the reduction.}
Observe that $M$ and $N$ are 1-interleaved if and only if there exist dashed morphisms such that the below diagram is commutative for every literal $x_{i_s}$ and for every clause $c_i$ containing $x_{i_s}$:
\begin{center} 
\begin{equation}
\begin{tikzcd}
M_{(x_{i_s},5)} = 0&0=N_{(x_{i_s},5)} & & M_{(c_i,5)}=0& 0=N_{(c_i,5)}\\
M_{(x_{i_s},4)} =  \kk\ar[ur, dashed]\ar[rrr, bend left=20]\ar[u] & 0=N_{(x_{i_s},4)} \ar[u]\ar[rrr, bend left=20]\ar[ul, dashed]  & & M_{(c_i,4)}=\kk\ar[ur, dashed]\ar[u] & 0=N_{(c_i,4)}\ar[u]\ar[ul, dashed]  \\
M_{(x_{i_s},3)} = \kk\ar[ur, dashed]\ar[rrr, bend left=20]\ar[u] & \kk^2=N_{(x_{i_s},3)}\ar[u]\ar[rrr, swap, bend left=20, "H_s"]\ar[ul, dashed]  & & M_{(c_i,3)}=\kk\ar[ur, dashed]\ar[u] & \kk^3=N_{(c_i,3)}\ar[u]\ar[ul, dashed] \\
M_{(x_{i_s},2)} = \kk\ar[ur, dashed]\ar[u] & \kk=N_{(x_{i_s},2)}\ar[u, swap,"(1;1)"]\ar[ul, dashed]  && M_{(c_i,2)}=\kk\ar[ur, dashed]\ar[u] & \kk^3=N_{(c_i,2)}\ar[u, swap,"I_3"] \ar[ul, dashed] \\
M_{(x_{i_s},1)} = \kk\ar[ur, dashed]\ar[u] & \kk =N_{(x_{i_s},1)}\ar[u]\ar[ul, dashed] & & M_{(c_i,1)}=0\ar[u] & 0=N_{(c_i,1)} \ar[u]
\end{tikzcd}
\label{eq:diagramProof}
\end{equation}
\end{center}
We shall see there are few degrees of freedom in the choice of interleaving morphisms. Indeed, consider the left part of the above diagram:
$$
\begin{tikzcd}
M & 0\ar[r] & \kk\ar[r, "1"] \ar[dr, thick]& \kk\ar[r, "1"]\ar[dr, thick] & \kk\ar[r, "1"]\ar[dr, thick] & \kk\ar[r, "1"]\ar[dr, thick] & 0 \\
N & 0\ar[r]& \kk\ar[r, swap, "1"]\ar[ur, thick]& \kk \ar[r, swap, "(1;1)"]\ar[ur, thick] & \kk^2 \ar[r]\ar[ur, dashed] & 0 \ar[r]\ar[ur, thick] & 0
\end{tikzcd}
$$
We leave it to the reader to verify that if $M$ and $N$ are 1-interleaved, then all the solid diagonal morphisms in the above diagram are completely determined by commutativity. For the dashed morphism $(\phi_{x_{i_s}}, \phi_{\neg x_{i_s}}): \kk^2 \to \kk$ there are two choices: by commutativity it must satisfy $(\phi_{x_{i_s}}, \phi_{\neg x_{i_s}})\cdot (1;1)= 1$ and thus $\phi_{x_{i_s}} + \phi_{\neg x_{i_s}}=1$. As $\kk = \Z/2\Z$, this implies that precisely one of $\phi_{x_{i_s}}$ and $\phi_{\neg x_{i_s}}$ is multiplication by 1. This corresponds to a choice of truth value for $x_{i_s}$: $\phi_{x_{i_s}} = 1 \iff x_{i_s}={\rm True}$ and $ \phi_{\neg x_{i_s}} = 1 \iff x_{i_s}={\rm False}$. Next, consider the right part of \ref{eq:diagramProof}:
$$
\begin{tikzcd}
M & 0\ar[r] &\kk \ar[r, "1"]\ar[dr, dashed]  &  \kk\ar[r, "1"]\ar[dr, thick]  & \kk \ar[r]\ar[dr, thick]  & 0\\
N & 0\ar[r] & \kk^3 \ar[r, swap, "1"]\ar[ur, dashed] & \kk^3\ar[r]\ar[ur, dashed]\ar[ur, dashed] & \ar[r]\ar[ur, thick] 0 & 0
\end{tikzcd}
$$
There are three non-trivial morphisms, out of which two are equal by commutativity. Let $Z_1^i: \kk\to \kk^3$ and $Z_2^i: \kk^3 \to \kk$ denote the two unspecified morphisms.  Returning to (\ref{eq:diagramProof}), we see that $Z_2^i$ must satisfy the following for $s\in \{1,2,3\}$:
$$
\begin{tikzcd}
\kk\ar[rr, "1"] & & \kk &  \\
& \kk^2\ar[ul, "{(\phi_{x_{i_s}},\phi_{\neg x_{i_s}})}"]\ar[rr, "H_s"]  & & \kk^3\ar[ul, dashed, swap, "Z_2^i"]
\end{tikzcd}
$$
Thus, $Z^i_2$ restricted to its $s$-th component equals either $\phi_{x_{i_s}}$ or $\phi_{\neg x_{i_s}}$, depending on whether $x_{i_s}$ or its negation $\neg x_{i_s}$ appears in the clause $c_i$. This implies that $Z_2^i$ is given by
$$
Z_2^i=
\begin{bmatrix}
\phi_{z_{i_1}} &  \phi_{z_{i_2}} & \phi_{z_{i_3}}
\end{bmatrix}
$$
Hence, if $M$ and $N$ are to be 1-interleaved, then there are no degrees of freedom in choosing $Z_2^i$ after the $\phi_{x_{i_s}}$ are specified. However, $Z_1^i $ only needs to satisfy  $Z_2^i\circ Z_1^i = 1.$ As this is the sole restriction imposed on $Z_1^i$, we see that this can be satisfied if and only if $Z_2^i \neq 0$, which is true if and only if $z_{i_s} = {\rm True}$ for at least one $s\in\{1,2,3\}$. 
\begin{theorem}
\label{teo:NPHard}
Let $\psi$ be a boolean formula as above. Then $\psi$ is satisfiable if and only if the associated persistence modules $M, N: \ZCat^{L\to C} \to \vect$ are 1-interleaved. 
\end{theorem}
\begin{proof}
Summarizing the above: we have that $M$ and $N$ are 1-interleaved if and only if we can choose morphisms $(\phi_{x_{i_s}}, \phi_{\neg x_{i_s}})$ such that $Z_2^i \neq 0$ for all clauses $c_i$. This means precisely that we can choose truth values for each $x_{i_s}$ such that every clause $c_i = z_{i_1} \vee z_{i_2} \vee z_{i_3}$ evaluates to true. This shows that a $1$-interleaving implies that $\psi$ is satisfiable. Conversely, if $\psi$ is satisfiable, then we see that the morphisms defined by $\phi_{x_{i_s}} = 1 \iff x_{i_s} = {\rm True}$ and $\phi_{\neg x_{i_s}} = 1 \iff x_i={\rm False}$ satisfy $Z_2^i\neq 0$ for every clause $c_i$. Thus, $M$ and $N$ are $1$-interleaved. 
\end{proof}
\begin{remark}
Let $i:\PCat \hookrightarrow \QCat$ be an inclusion of posets and $M: \PCat\to \Vect$. There are multiple functorial ways of extending $M$ to a representation $E(M): \QCat\to \Vect$, e.g. by means of left or right Kan extensions. This is a key ingredient in one of the more recent proofs of \cref{teo:isometry}; see \cite{bubenik2017higher} for details. However, if we impose the condition that $E(M)\circ i \cong M$ then such an extension need not exist. Indeed, \cref{teo:NPHard} implies that the associated decision problem is NP-complete. 
\end{remark}

\section{Interleavings of Multidimensional Persistence Modules}
\label{sec:multid}
Recall that a constrained invertibility (CI) problem is a triple $(P,Q,n)$ where $P$ and $Q$ are subsets of $\{1,2, \dots, n\}^2$, and that a CI-problem is solvable if there exists an invertible $n \times n$ matrix $M$ such that $M_{(i,j)} = 0$ for all $(i,j) \in P$ and $M^{-1}_{(i',j')} = 0$ for all $(i',j') \in Q$. We shall show that a CI-problem is solvable if and only if a pair of associated persistence modules $\ZCat^2\to \Vect$ is 1-interleaved. Hence, if deciding solvability is NP-hard, then so is computing the interleaving distance for multidimensional persistence modules.
\begin{example}
\label{ex:PQ}
Let $P = \{(2,2), (3,3) \}, Q = \{(2,3), (3,2) \} \subset \{1,2,3\}^2$. Then $(P,Q,3)$ is solvable by
$$
M =
\begin{bmatrix}
1&1&1\\
1&0&1\\
1&1&0
\end{bmatrix},
\quad
M^{-1} =
\begin{bmatrix}
-1&1&1\\
1&-1&0\\
1&0&-1
\end{bmatrix}.
$$
\end{example}
\begin{example}
Let $P = \{(1,1), (1,3) \}, Q = \{(2,1) \} \subset \{1,2,3\}^2$. Then $(P,Q,3)$ is not solvable, as $(MN)_{(1,1)} = 0$ for all $3 \times 3$-matrices $M,N$ with $M_{(1,1)} = M_{(1,3)} = N_{(2,1)} = 0$. Note that it matters that we view $P$ and $Q$ as subsets of $\{1,2,3\}^2$ and not of $\{1, \dots, n\}^2$ for some $n > 3$, in which case $(P,Q)$ would be solvable.
\end{example}

\begin{example}
Observe that a CI-problem $(P, \emptyset, n)$ reduces to a bipartite matching problem. Build a graph $G$ on $2n$ vertices $\{v_1, \ldots, v_n, u_1, \ldots, u_n\}$ with an edge from $v_i$ to $u_j$ if $(i,j)\notin P$.  Then the CI-problem is solvable if and only if there exists a perfect matching of $G$. 
\end{example}
A CI-problem can be seen as a problem of choosing weights for the edges in a directed simple graph: Given $(P,Q,n)$, let $G$ be the bipartite directed simple graph with vertices $\{u_1, \dots, u_n, v_1, \dots, v_n\}$, an edge from $u_i$ to $v_j$ if $(i,j) \notin P$, and an edge from $v_j$ to $u_i$ if $(j,i) \notin Q$. Solving $(P,Q,n)$ is then equivalent to weighting the edges in $G$ with elements from $\kk$ so that $$\sum_{j=1}^{n} w(u_i,v_j) w(v_j,u_i) = 1$$ for all $i$, and $$\sum_{j=1}^{n} w(u_i,v_j) w(v_j,u_{i'}) = 0$$ for all $i \neq i'$, where $w(u,v)$ is the weight of the edge from $u$ to $v$ if there is one, and $0$ if not. If the weights are elements of $\Z/2\Z$, this is equivalent to picking a subset of the edges such that there is an odd number of paths of length two from any vertex to itself and an even number of paths of length two from any vertex to any other vertex.

Fix a CI-problem $(P,Q,n)$ and let $m = |P| + |Q|$. We will construct $\ZCat^2$-indexed modules $M$ and $N$ that are $1$-interleaved if and only if $(P,Q,n)$ is solvable, and that are zero outside a grid of size $(2m+3) \times (2m+3)$ in $\ZCat^2$. The dimension of each vector space $M_{(a,b)}$ or $N_{(a,b)}$ is bounded by $n$, so the total dimensions of $M$ and $N$ are polynomial in $n$.

For $p \in \ZCat^2$, let $\langle p \rangle = \{ q \in \ZCat^2 \mid p \leq q \leq (2m+2,2m+2) \}$. Let $\W$ be the interval $\bigcup_{k=0}^m \langle (2m-2k,2k) \rangle$, and for $i \in \{1, 2, \dots, m\}$, let $x_i = (2m-2i+1, 2i-1)$; see \cref{Wsupport}.

Write $P = \{(p_1,q_1), \dots, (p_r,q_r) \}$ and $Q = \{(p_{r+1},q_{r+1}),\dots, (p_m,q_m) \}$. We define $M = \bigoplus_{i=1}^n I^{\I_i}$ and $N = \bigoplus_{i=1}^n I^{\J_i}$, where $\I_i$ and $\J_i$ are constructed as follows: let $\I_i^0 = \J_i^0 = \W$ for all $i$. For $k = 1, 2, \dots, r$, let

\begin{minipage}{.5\linewidth}
\[
\I_{i}^k =
\begin{cases}
\I_{i}^{k-1} \cup \langle x_k - (1,1)\rangle, & \text{ if } i=p_k \\
\I_i^{k-1} \cup \langle x_k \rangle, & \text{ if } i\neq p_k
\end{cases},
\]
\end{minipage}%
\begin{minipage}{.4\linewidth}
\[\J_{i}^k =
\begin{cases}
\J_{i}^{k-1}, & \text{ if } i=q_k \\
\J_i^{k-1} \cup \langle x_k \rangle, & \text{ if } i\neq q_k
\end{cases}
\]
\end{minipage}\\

\noindent and for $k = r+1, \dots, m$, let

\begin{minipage}{.45\linewidth}
\[
\I_{i}^k =
\begin{cases}
\I_{i}^{k-1}, & \text{ if } i=q_k \\
\I_i^{k-1} \cup \langle x_k \rangle, & \text{ if } i\neq q_k
\end{cases},
\]
\end{minipage}%
\begin{minipage}{.4\linewidth}
\[
\J_{i}^k =
\begin{cases}
\J_{i}^{k-1} \cup \langle x_k - (1,1) \rangle, & \text{ if } i=p_k \\
\J_i^{k-1} \cup \langle x_k \rangle, & \text{ if } i\neq p_k
\end{cases}
\]
\end{minipage}\\
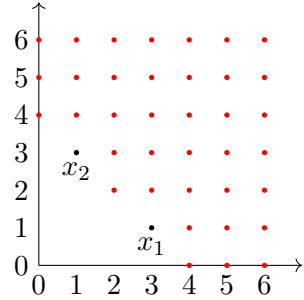
\begin{wrapfigure}[12]{r}{0.35\textwidth}
\centering
\begin{tikzpicture}[scale=.5]
\node[below] at (3,1){$x_1$};
\node[below] at (1,3){$x_2$};
\fill (3,1) circle [radius=2pt];
\fill (1,3) circle [radius=2pt];
\draw[->] (0,0) to (7,0);
\draw[->] (0,0) to (0,7);
\foreach \i in {0,...,6}{
\node[below] at (\i,0){$\i$};
\node[left] at (0,\i){$\i$};
\foreach \j in {4,...,6}{
\fill[red] (\i,\j) circle [radius=2pt];
}}
\foreach \i in {2,...,6}{
\foreach \j in {2,...,3}{
\fill[red] (\i,\j) circle [radius=2pt];
}}
\foreach \i in {4,...,6}{
\foreach \j in {0,1}{
\fill[red] (\i,\j) circle [radius=2pt];
}}
\end{tikzpicture}
\caption{The interval $\W$ for $m=2$ along with $x_1 = (3,1)$ and $x_2 = (1,3)$. \label{Wsupportx}}
\label{Wsupport}
\end{wrapfigure}
\noindent and let $\I_i = \I_i^m$ and $\J_i = \J_i^m$. This way, we ensure that there is no nonzero morphism from $I^{\I_i}$ to $I^{\J_j}(1)$ when $(i,j) \in P$, and no nonzero morphism from $I^{\J_j}$ to $I^{\I_i}(1)$ when $(j,i) \in Q$. In all other cases, there exist nonzero morphisms.
\begin{lemma}
Suppose $(i,j) \notin P$. Then there is an isomorphism $\Hom(I^{I_i}, I^{\J_j}(1)) \cong \kk$. In particular, any morphism $f\in\Hom(I^{I_i}, I^{\J_j}(1))$ is completely determined by $f_{(2m+1,2m+1)}$: if $f_p$ is nonzero, then $f_p = f_{(2m+1,2m+1)}$.
\end{lemma}
The same holds if $(j,i) \notin Q$ instead of $(i,j) \notin P$, and $\I_i$ and $\J_j$ are interchanged. As $f_p$ is a $\kk$-endomorphism, this implies that any $f$ can be identified with an element of $\kk$. 
\begin{proof}
Let $f: I^{\I_i} \to I^{\J_j}(1)$ be nonzero. If $p \notin \I_i$ or $p \nleq (2m+1,2m+1)$, $f_p = 0$. For $(2m+1,2m+1) \geq p \in \I_i$, we have $p+(1,1) \in \J_i$ by construction and the fact that $(i,j) \notin P$, so $\phi_{I^{\J_j}}(p+(1,1), (2m+2,2m+2))$ is nonzero and hence the identity. We get
\begin{align*}
f_p &= \phi_{I^{\J_j}(1)}(p, (2m+1,2m+1)) \circ f_p\\
&= f_{(2m+1,2m+1)} \circ \phi_{I^{\I_i}}(p, (2m+1,2m+1))= f_{(2m+1,2m+1)}.
\end{align*}
\end{proof}
Describing a morphism from $M = \bigoplus_{i=1}^n I^{\I_i}$ to $N(1) = \bigoplus_{j=1}^n I^{\J_j}(1)$ is the same as describing morphisms from $I^{\I_i}$ to $I^{\J_j}(1)$ for all $i$ and $j$\footnote{$\Hom(\oplus_i M_i, \oplus_j N_j) \cong \oplus_i \oplus_j \Hom(M_i, N_j)$.}. We have just proved that these can be identified with elements of $\kk$, so we conclude that any $f: M \to N(1)$ is uniquely defined by an $n \times n$-matrix $A_f$ where the entry $(i,j)$ is the element in $\kk$ corresponding to the morphism $I^{\I_i} \to I^{\J_j}(1)$ given by $f$. Note that we get the same result by writing $f_{(2m,2m)} = f_{(2m+1,2m+1)}: \kk^n \to \kk^n$ as a matrix, where each copy of $\kk$ in the domain and codomain comes from one of the interval modules $I^{\I_i}$ or $I^{\J_j}(1)$, respectively.

If we also have a morphism $g: N \to M(1)$, we can define a matrix $A_g$ symmetrically, and similarly $A_g$ is $g_{(2m,2m)} = g_{(2m+1,2m+1)}: \kk^n \to \kk^n$ in matrix form.
\begin{theorem}
\label{thm:CI=int}
With $f$ and $g$ as above, $(f,g)$ is a $1$-interleaving if and only if $A_f$ and $A_g$ are inverse matrices.
\end{theorem}
\begin{proof}
Suppose $(f,g)$ is a $1$-interleaving. The internal morphism $\phi_M((2m,2m), (2m+2,2m+2))$ is the identity on $\kk^n$, and is by definition of interleaving the same as $$g(1)_{(2m,2m)} \circ f_{(2m,2m)} = g_{(2m+1,2m+1)} \circ f_{(2m,2m)} = A_g A_f.$$ Thus $A_g A_f$ is the identity matrix, and so $A_f$ and $A_g$ are inverses of each other, as both are $n \times n$-matrices.

Suppose $A_f$ and $A_g$ are inverse matrices. We must check that at every point $p \in \mathbb{Z}^2$, $\phi_M(p, p + (2,2)) = g(1)_p \circ f_p$. If $p \nleq (2m,2m)$ or $M(p) = 0$, both sides are zero. If $p \leq (2m,2m)$ and $M(p) \neq 0$, $$g(1)_p \circ f_p = \phi_M(p + (2,2), (2m+2,2m+2)) \circ g(1)_p \circ f_p,$$ since $\phi_M(p + (2,2), (2m+2,2m+2))$ must be the identity by construction of $M$ and the fact that $M(p) \neq 0$. This is equal to
\begin{align*}
&\phi_M(p + (2,2), (2m+2,2m+2)) \circ g_{p+(1,1)} \circ f_p\\
&=g_{(2m+1,2m+1)} \circ \phi_N(p + (1,1), (2m+1,2m+1)) \circ f_p\\
&=g_{(2m+1,2m+1)} \circ f_{(2m,2m)} \circ \phi_M(p, (2m,2m))\\
&=A_g A_f \circ \phi_M(p, (2m,2m))= \phi_M(p, (2m,2m))=\phi_M(p, p + (2,2)).
\end{align*}
\end{proof}
We have proved that defining morphisms $f: M \to N(1)$ and $g: N \to M(1)$ is the same as choosing $n \times n$-matrices $A_f$ and $A_g$ such that the entries corresponding to the elements of $P$ and $Q$ are zero, and that $(f,g)$ is a $1$-interleaving if and only if $A_f$ and $A_g$ are inverse matrices. Thus $M$ and $N$ are $1$-interleaved if and only if the CI-problem $(P,Q,n)$ is solvable.

We constructed $M$ and $N$ by setting all the interval modules comprising $M$ and $N$ equal to $I^\W$, then modifying them in $m$ steps each, where the complexity of each step is clearly polynomial in $n$. Thus the complexity of constructing $M$ and $N$ is polynomial in $n$, and so are the total dimensions of $M$ and $N$. Taking $n^2$ as the input complexity of solving a CI-problem $(P,Q,n)$, we have proved a reduction implying the following theorem:
\begin{theorem}
Determining the interleaving distance for modules $\ZCat^2\to \Vect$ is CI-hard. 
\end{theorem}
\begin{remark}
We give an example in \cref{app:unstable} of a CI-problem whose associated matrices $M$ and $N$ satisfy $d_I(M,N) = 1$ and $d_B(\B(M),\B(N)) = 2$. This shows that it is not enough to find the bottleneck distance of the barcodes of $M$ and $N$ to decide whether $M$ and $N$ are $1$-interleaved and thus whether the CI-problem is solvable. In fact, recent work shows that $d_B$ can be efficiently computed \cite{dey2018computing}. 
\end{remark}

We end this paper with the somewhat surprising observation that the interleaving distance of the above interval decomposable modules depends the characteristic $\textup{char}(\kk)$ of the underlying field $\kk$. That is, let $M, N: \ZCat^2\to \Vect_\kk$, $M', N': \ZCat^2\to \Vect_{\kk'}$, $\kk\neq \kk'$, and for which $\B(M) = \B(M')$ and $\B(N) = \B(N')$. Clearly, any matching distance $d$ would satisfy $d(\B(M), \B(N)) = d(\B(M'), \B(N'))$, but it is not always true that $d_I(M, N) = d_I(M', N')$.

For a fixed $n \geq 2$, let $Q = \{(2,2), \dots, (n+2,n+2) \}$ and $P = \{(1,1) \} \cup \{2,\dots,n+2\}^2 \setminus Q$. Then the CI-problem $(P,Q,n+2)$ is solvable if and only if the characteristic of $\kk$ divides $n$. We will only prove this for $n=2$ for clarity, but the argument easily generalizes to all $n$.

Assume that $(M,M^{-1})$ is a solution to $(P,Q,4)$:
$$
M =
\begin{bmatrix}
0&?&?&?\\
a&?&0&0\\
b&0&?&0\\
c&0&0&?
\end{bmatrix},
\quad
M^{-1} =
\begin{bmatrix}
?&d&e&f\\
?&0&?&?\\
?&?&0&?\\
?&?&?&0
\end{bmatrix}
$$
Here we have put the entries corresponding to the elements of $P$ and $Q$ equal to $0$, and left the rest as unknown. The entries we will use in the calculations that follow are labeled $a,b,c,d,e,f$. We see that $(MM^{-1})_{(2,2)} = ad$, $(MM^{-1})_{(3,3)} = be$, $(MM^{-1})_{(4,4)} = cf$, that is, $ad = be = cf = 1$. At the same time, $(M^{-1}M)_{(1,1)} = ad + be + cf$, so we get $1=1+1+1$, or $2=0$. Thus $\textup{char}(\kk) = 2$, and in this case we can put all the unknowns in $M$ and $M^{-1}$ above equal to $1$ to obtain a solution. (For $n>2$, we put the nonzero elements on the diagonal of $M$ equal to $-1$.)

Our motivation for introducing CI-problems was working towards determining the computational complexity of calculating the interleaving distance. While the last examples say little about complexity, they illustrate the underlying philosophy of our approach: By considering CI-problems, we can avoid the confusing geometric aspects of persistence modules and interleavings. E.g., in the case above, working with persistence modules over a $23 \times 23$ size grid is reduced to looking at a pair of $4 \times 4$-matrices.

\section{Discussion}
The problem of determining the computational complexity of computing the interleaving distance for multidimensional persistence modules (valued in $\Vect$) was first brought up in Lesnick's thesis \cite{lesnick2012multidimensional}. Although it has been an important open question for several years, a non-trivial lower bound on the complexity class has not yet been given. In light of \cref{thm:brooksbank}, one might hope that tools from computational algebra can be efficiently extended to the setting of interleavings. Theorem \ref{thm:NPHARD} is an argument against this, as it shows that the problem of computing the interleaving distance is NP-hard in general. This leads us to conjecture that the problem of computing the interleaving distance for multidimensional persistence modules is also NP-hard. Unfortunately, writing down the conditions for an interleaving becomes intractable already for small grids. To make the decision problem more accessible to researchers in other fields of mathematics and computer science, we have shown that the problem is at least as hard as an easy to state matrix invertibility problem. We speculate that this problem is also NP-hard. If that is not the case, then an algorithm would provide valuable insight into the interleaving problem for interval decomposable modules.

\appendix

\section{Bottleneck Distance}
\label{sec:bottleneck}
A \emph{matching} $\sigma$ between multisets $S$ and $T$ (written as $\sigma: S \nrightarrow T$) is a bijection $\sigma: S\supseteq S^\prime \to T^\prime\subseteq T$. Formally, we regard $\sigma$ as a relation $\sigma\subseteq S\times T$ where $(s,t)\in \sigma$ if and only if $s\in S^\prime$ and $\sigma(s) = t$. We call $S^\prime$ and $T^\prime$ the \emph{coimage} and \emph{image} of $\sigma$, respectively, and denote them by $\coim\sigma$ and $\im\sigma$.  If $w\in \coim \sigma\cup \im \sigma$, we say that \emph{$\sigma$ matches $w$.}  

We say intervals $\J,\K\subseteq \ZCat^n$ are $\delta$-interleaved if $I^\J$ and $I^\K$ are $\delta$-interleaved. Similarly, we say $\J$ is $2\delta$-trivial if $I^\J$ is $\delta$-interleaved with the $0$-module, i.e. the module $I^\emptyset$. For $\C$ a barcode over $\ZCat^n$ and $\delta\geq 0$, define $\C_{\delta}\subseteq \C$ to be the multiset of intervals in $\C$ that are not $\delta$-trivial.  

Define a \emph{$\delta$-matching} between barcodes $\mathcal C$ and $\mathcal D$ to be a matching 
$\sigma:\C\nrightarrow \D$ satisfying the following properties:
\begin{enumerate}
\item $\C_{2\delta} \subseteq \coim \sigma$ and $\D_{2\delta}\subseteq \im \sigma$.
\item If $\sigma(\J)=\K$, then $\J$ and $\K$ are $\delta$-interleaved.
\end{enumerate}
For barcodes $\C$ and $\D$, we define the bottleneck distance $d_B$ by
\[d_B(\mathcal C,\mathcal D)=\min\, \{\delta\in \{0,1, 2, \ldots\} \mid \exists\textup{ a }\delta\textup{-matching between }\mathcal C\textup{ and }\mathcal D\}.\]
It is not hard to check that $d_B$ is an extended pseudometric.  In particular, it satisfies the triangle inequality.
\section{Discrete Modules}
\label{app:B}
We define an \emph{(injective) $n$-D grid} to be a function $\G:\Z^n\to \R^n$ given by \[\G(z_1,\ldots, z_n)=(\G_1(z_1), \ldots, \G_n(z_n))\] for strictly increasing functions $\G_i:\Z\to \R$ with $\lim_{i\to -\infty}=-\infty$ and $\lim_{i\to \infty}=\infty$.  

Define $\fl_\G:\R^n\to \im(\G)$ by $\fl_\G(t)=\max \{s\in \im(\G)\mid s\leq t\}.$

For $\G$ an $n$-D grid, we let $\CoEx_\G:\C^{\ZCat^n}\to \C^{\RCat^n}$:
\begin{enumerate}
\item For $M$ a $\ZCat^n$-indexed persistence module and $a,b\in \RCat^n$, \[\CoEx_\G(M)_a=M_y,\qquad  \varphi_{\CoEx_\G(M)}(a,b)=\varphi_M(y,z),\] where $y,z\in \ZCat^2$ are given by $\G(y)=\fl_\G(a)$ and $\G(z)=\fl_\G(b)$.
\item The action of $\CoEx_\G$ on morphisms is the obvious one.
\end{enumerate}
Let \[(-)|_\G:\C^{\RCat^n}\to \C^{\ZCat^n}\] denote the restriction along $\G$.

We say that $M: \RCat^n\to \C$ is \emph{discrete} if there exists an $n$-grid $\G$ such that $M \cong \CoEx_\G(M|_\G)$. Clearly, if $M$ and $N$ are discrete then we may choose a grid $\G$ such that $M \cong \CoEx_\G(M|_\G)$ and $N \cong \CoEx_\G(N|_\G)$.

\section{Interleavings of Functors $\ZCat\to \Vect$}
\label{app:vec}
It is well-known \cite{carlsson2009theory} and easy to see that a persistence module $M: \ZCat\to \Vect$ is completely determined by its associated \emph{rank invariant} ${\rm rk}_M$, $${\rm rk}_M(a,b) = {\rm rank}(\phi_M(a,b)), \qquad a\leq b \in \ZCat.$$ The rank of an $m_1 \times m_2 $-matrix can be calculated in $\Op(m_1m_2^{\omega-1})$ \cite{ibarra1982generalization}, where $\omega$ is the matrix multiplication exponent. Let $d_i = \dim M_i$ and $d=\dim M = \sum_i d_i$, and assume that we are given a list of all $i$ such that $M_i$ is nonzero. The cost of calculating ${\rm rk}_M(i,j)$ for all pairs $i<j$ in the list is at most
\begin{align*}
\sum_{i<j} Cd_i d_j^{\omega-1} \leq C\left(\sum_i d_i \right) \left(\sum_i d_i^{\omega-1} \right)\leq C\left(\sum_i d_i \right) \left(\sum_i d_i \right)^{\omega-1} &\leq C\left(\sum_i d_i \right)^\omega\\ &= Cd^\omega
\end{align*}
for a sufficiently large constant $C$. This shows that the complexity of computing ${\rm rk}_M$ is $\Op(d^\omega)$. Note that the number of intervals $[a,b]$ in the barcode $\B(M)$ is ${\rm rk}_M(a,b) - {\rm rk}_M(a-1,b) - {\rm rk}_M(a,b+1) + {\rm rk}_M(a-1,b+1)$. Thus, once we got the rank invariant, we can extract $\B(M)$ in $\Op(d^2)$ operations. In conclusion, we have provided an algorithm which computes $\B(M)$ from $M$ in $\Op(d^\omega+d^2) = \Op(d^\omega)$ operations. 

Observe that $|\B(M)| \leq \dim M$. Now, assume that we are given barcodes $\B(M)$ and $\B(N)$, with $n=\dim M + \dim N$, and we want to decide if $M$ and $N$ are $\delta$-interleaved. By \cref{teo:isometry}, this is equivalent to deciding if there is a $\delta$-matching between $\B(M)$ and $\B(N)$, which can be done in $\Op(b^{1.5} \log b)$, where $b = |\B(M)| + |\B(N)| \leq n$ \cite{kerber2017geometry}. Thus, we can decide if $M,N: \ZCat\to \Vect$ are $\delta$-interleaved in $\Op(n^\omega)$.

\section{The Isomorphism Problem for $\ZCat^2 \to \Set$}
\label{app:isoGI}
The isomorphism problem for Reeb graphs can be rephrased as an isomorphism problem of $\ZCat^2$-indexed persistence modules. Indeed, following \cite{botnan2016algebraic}, one sees that Reeb graphs can be viewed as functors $\RCat^2 \to \Set$, and by \cite{de2015categorified} it follows that these functors are discrete in the sense of \cref{app:discrete}. As the isomorphism problem for Reeb graphs is graph isomorphism hard \cite{de2015categorified}, it follows immediately that the same is true for modules ${\ZCat^2}\to \Set$. We shall show that these problems are in fact graph isomorphism \emph{complete}. Since we have chosen the total cardinality as the input size, and every functor $\ZCat^2\to \Set$, except the one sending everything to the empty set, has infinite total cardinality, we consider functors $[n]^2 \to \Set$ instead. Here $[n]^2$ is $\{1,2,\dots,n\}^2$ considered as a full subcategory of $\ZCat^2$. 

Let $M, N: [n]^2\to \Set$. We shall associate a pair of multigraphs to $M$ and $N$ in a way that ensures that $M$ and $N$ are isomorphic if and only if the associated multigraphs are isomorphic. The isomorphism problem for multigraphs is GI-complete \cite{zemlyachenko1985graph}. 

An isomorphism between $M,N: [n]^2 \to \Set$ is a natural isomorphism, i.e. a natural transformation with a two-sided inverse. Concretely, such an isomorphism $f$ consists of bijections $f_p: M_p \to N_p$ for all $p \in [n]^2$ that commute with the internal morphisms of $M$ and $N$, meaning that $f_{p+(0,1)} \circ \phi_M(p,p+(0,1)) = \phi_N(p,p+(0,1)) \circ f_p$ and $f_{p+(1,0)} \circ \phi_M(p,p+(1,0)) = \phi_N(p,p+(1,0)) \circ f_p$ hold whenever everything is defined. It is not hard to check that $f^{-1}$ defined by $\left(f^{-1}\right)_p = \left(f_p\right)^{-1}$ is an inverse of $f$.

Given modules $M, N: [n]^2 \to \Set$, we may assume that their pointwise cardinalities are the same, since if not, we can immediately conclude that they are not isomorphic. Let $c = |M| = |N|$. We also assume that $M_p$ and $N_p$ are nonempty on $p=(1,1)$, and for at least one $p \in \{1\} \times [n] \cup [n] \times \{1\}$. This implies $c \geq n$. We define the graph $G(M) = (V,E)$ as follows.
\begin{itemize}
\item $V = \bigcup_{p \in [n]^2} M_p \cup \{T\}$.
\item There is a single edge between $x \in M_p$ and $y \in M_q$ if $\phi_M(p,q)(x) = y$ and either $q = p + (0,1)$ or $q = p + (1,0)$.
\item For $x \in M_{(a,b)}$, there are $n(a-1)+b$ edges between $x$ and $T$.
\end{itemize}
Except from the ones described, there are no edges in $G(M)$. We can visualize $G(M)$ as the graph we get by putting $|M_p|$ vertices at each point in $[n]^2$ and short horizontal and vertical edges given by the internal morphisms of $M$, and in addition one vertex $T$ which is incident to a certain number of edges from each other vertex. We have $|V| = c+1$ and $|E| \leq 2c^2 + c n^2 \leq (2+c)c^2$, since at most $2c^2$ edges come from the internal morphisms of $M$ and $cn^2$ is an upper bound on the number of edges incident to $T$. In other words $|V| + |E|$ is polynomial in $c$. Defining $G(N) = (V',E')$ analogously with $T'$ in place for $T$, we get the same for $|V'| + |E'|$.

Now we consider what an isomorphism $f:V \to V'$ from $G(M)$ to $G(N)$ must look like. Except for cases with $c \leq 2$, both graphs have exactly one vertex that is adjacent to all other vertices, so $T$ must be sent to $T' \in V'$. Since there are $n(a-1)+b$ edges between $x \in M_{(a,b)}$ and $T$, there must be $n(a-1)+b$ edges between $f(x)$ and $f(T) = T'$, implying $f(x) \in N_{(a,b)}$. Thus the restriction of $f$ to $M_p$ is a bijection $M_p \to N_p$ for each $p \in [n]^2$.

It is easy to see that $f$ is functorial. That is, there is an edge between $x \in M_p$ and $y \in M_q$ if and only if there is an edge between $f(x) \in N_p$ and $f(y) \in N_q$. Hence, we conclude that $f$ defines an isomorphism between $M$ and $N$ in the obvious way. 


\begin{remark}
With small adjustments, the reduction from isomorphism of functors $[n]^2 \to \Set$ to isomorphism of multigraphs would work just as well for any poset category $\PCat$ in place of $[n]^2$. This shows that  determining isomorphism between $\Set$-valued functors is at most as hard as GI regardless of the poset category.
\end{remark}

\section{The Isomorphism Problem for $\ZCat \to \Set$}
\label{app:mergetree}
We consider functors $[n] \to \Set$, where $[n] = \{1,2,\dots,n\}$ is a subcategory of $\ZCat$, as in \cref{app:isoGI}. A rooted tree is a tree with one vertex chosen as the root, and an isomorphism between two rooted trees is a graph isomorphism that sends the root of one tree to the root of the other. We will show that deciding whether functors $[n] \to \Set$ are isomorphic is linear in the total cardinality by reducing it to checking isomorphism between rooted trees, which is known to be linear in the number of vertices \cite[p.~85]{aho1974design}.

Given $M: [n]\to \Set$, let $T(M)$ be the rooted tree with vertex set $\bigcup_{k=1}^n M_k \sqcup \{r\}$, where we choose $r$ as the root and there is an edge between $x \in M_k$ and $y \in M_{k+1}$ if $\phi_M(k,k+1)(x) = y$. For persistence modules $M$ and $N$,  an isomorphism between $T(M)$ and $T(N)$ is a function that sends the root of $T(M)$ to the root of $T(N)$ and restricts to a bijection from $M_k$ to $N_k$ for each $k$. Moreover, $f$ preserves parent-child relations, which means that for $x \in M_k$, $k < n$, $\phi_N(k,k+1)(f(x)) = f(\phi_M(k,k+1)(x))$. This is exactly what it takes for $f$ restricted to $\bigcup_{k=1}^n M_k$ to define a natural transformation from $M$ to $N$. Thus, $T(M)$ and $T(N)$ are isomorphic as rooted trees if and only if $M$ and $N$ are isomorphic as functors.

The number of vertices of $T(M)$ is one more than the total cardinality of $M$. Assuming that for each $k$ we are given a list of tuples $(x,\phi_M(k,k+1)(x))$, where $x$ runs through the elements of $M_k$, we have exactly the information needed to run the algorithm in \cite[p.~84]{aho1974design} for checking isomorphism of $T(M)$ and $T(N)$ in linear time. Thus, deciding whether merge trees are isomorphic can be done in time linear in $|M| + |N|$.

\section{Example $d_I \neq d_B$}
\label{app:unstable}
Consider the CI-problem $(P,Q,3)$, where $P = \{(2,3),(3,2) \}$ and $Q = \{(2,2),(3,3) \}$. Applying the algorithm in Section \ref{sec:multid}, we get modules $M = I^{\I_1} \oplus I^{\I_2} \oplus I^{\I_3}$ and $N = I^{\J_1} \oplus I^{\J_2} \oplus I^{\J_3}$, where
\begin{align*}
I_1 &= \W \cup \langle x_1 \rangle \cup \langle x_2 \rangle \cup \langle x_3 \rangle \cup \langle x_4 \rangle,\\
I_2 &= \W \cup \langle x_1 - (1,1) \rangle \cup \langle x_2 \rangle \cup \langle x_4 \rangle,\\
I_3 &= \W \cup \langle x_1 \rangle \cup \langle x_2 - (1,1) \rangle \cup \langle x_3 \rangle,\\
J_1 &= \W \cup \langle x_1 \rangle \cup \langle x_2 \rangle \cup \langle x_3 \rangle \cup \langle x_4 \rangle,\\
J_2 &= \W \cup \langle x_1 \rangle \cup \langle x_3 - (1,1) \rangle \cup \langle x_4 \rangle,\\
J_3 &= \W \cup \langle x_2 \rangle \cup \langle x_3 \rangle \cup \langle x_4 - (1,1) \rangle.
\end{align*}
By Example \ref{ex:PQ}, $(P,Q,3)$ is solvable, which means that $M$ and $N$ are $1$-interleaved. Since they are not isomorphic, $d_I(M,N) = 1$.

\begin{figure}
\centering
\begin{tikzpicture}[scale = 1.5]
\node (a) at (0,0){$J_1$};
\node (b) at (1.5,0){$J_2$};
\node (c) at (3,0){$J_3$};
\node (A) at (0,2){$I_1$};
\node (B) at (1.5,2){$I_2$};
\node (C) at (3,2){$I_3$};
\draw[<->, thick] (a) to (A);
\draw[<->, thick] (a) to (B);
\draw[<->, thick] (a) to (C);
\draw[<->, thick] (A) to (b);
\draw[<->, thick] (A) to (c);
\draw[->, thick] (B) to (b);
\draw[->, thick] (C) to (c);
\draw[->, thick] (b) to (C);
\draw[->, thick] (c) to (B);
\end{tikzpicture}
\caption{Graph illustrating possible nonzero morphisms between interval modules; see \cref{app:unstable}.}
\label{fig:matching}
\end{figure}
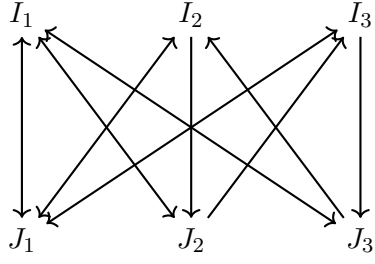

The graph in Figure \ref{fig:matching} has an edge from $A$ to $B$ if $A$ and $B$ are in different barcodes and there is a nonzero morphism from $\I^A$ to $\I^B(1)$. (A double-headed arrow means an edge in each direction.) In a $1$-matching between $\B(M)$ and $\B(N)$, if there is one, we need to match each $I_i$ with a $J_j$, and each corresponding pair of interval modules needs to be $1$-interleaved. Specifically, there needs to be a nonzero morphism both from $\I^{I_i}$ to $\I^{J_j}(1)$ and from $\I^{J_j}$ to $\I^{I_i}(1)$, that is, there must be edges in both directions between $I_i$ and $J_j$ in the graph. We see that both $I_2$ and $I_3$ can only be matched with $J_1$, and $J_1$ can only be matched with one of them. Thus there is no $1$-matching between $\B(M)$ and $\B(N)$. On the other hand, all the intervals are $2$-interleaved, so any bijection between $\B(M)$ and $\B(N)$ gives a $2$-matching. In other words, $d_B(\B(M), \B(N)) = 2$.
\bibliography{refs}

\end{document}